\documentclass[11pt,a4paper]{article}
\usepackage[english]{babel} 
\usepackage{amsmath}
\usepackage{mathtools}
\usepackage{amsfonts}
\usepackage{amssymb}
\usepackage{amsthm}
\usepackage{makeidx}
\usepackage{graphicx}
\usepackage{natbib}
\usepackage{setspace}
\usepackage{float}
\usepackage{hyperref}
\usepackage[margin=1in]{geometry}
\usepackage{setspace}
\usepackage[titletoc]{appendix}
\usepackage{cleveref}
\usepackage{cleveref}
\usepackage{systeme}
\allowdisplaybreaks

\hypersetup{
	colorlinks   = true,
	citecolor    = blue,
	urlcolor = blue
}

\newtheorem{Def}{Definition}
 \newtheorem{Thm}{Theorem}
 
 \newtheorem{Lmm}{Lemma}
 \newtheorem{Crl}{Corollary}
 \theoremstyle{definition}
 \newtheorem{Exm}{Example}
 
 \theoremstyle{remark}
 \newtheorem{Rmk}{Remark}

\DeclareMathOperator*{\argmin}{arg\,min}
\DeclareMathOperator*{\essinf}{ess\,inf}
\DeclareMathOperator*{\esssup}{ess\,sup}

\newcommand{\E}{{E}}

\title{Optimal hedging with variational preferences under convex risk measures}
\author{Marcelo Righi	\\
	\textit{Federal University of Rio Grande do Sul}\\
	marcelo.righi@ufrgs.br\footnote{We are grateful for the financial support of CNPq (Brazilian Research Council) projects number 302614/2021-4 and 401720/2023-3.}}
\date{}
\begin{document}
\maketitle
\begin{abstract}
We expose a theoretical hedging optimization framework with variational preferences under convex risk measures. We explore a general dual representation for the composition between risk measures and utilities. We study the properties of the optimization problem as a convex and monotone map per se. We also derive results for optimality and indifference pricing conditions. We also explore particular examples inside our setup. 

\end{abstract}

\newpage

\section{Introduction}\label{sec:intro}

In this note, we expose a theoretical hedging optimization framework with variational preferences under convex risk measures. Option hedging is a main issue in the field of mathematical finance, tracing back
to the seminal papers of \cite{Black1973} and \cite{Merton1973}. 
Consider a market with $n\in\mathbb{N}$ (see formal mathematical definitions below) assets $S=(S_1,\dots,S_n)\subset L^p_+$. Each of these assets is traded at time $0$ under known prices. We refer to $\Delta S=(\Delta S_1,\dots,\Delta S_n)\subset L^p$ their (discounted) price variations. We then have some contingent claim $H\in L^p$ to be hedged, i.e. the agent sells $H$ and buys some portfolio $V$, ending up with the hedged position $V-H$. We fix our hedging cost to be a scalar $V_0>0$, and the financial positions generated with such setup as \[\mathcal{V}_{V_0}=\left\lbrace V=V_0+\sum_{i=1}^n h_i\Delta S_i\colon h\in\mathbb{R}^n\right\rbrace.\] 

We say that $\mathcal{V}_{V_0}$ is the unconstrained set of such trading
strategies. The completeness property, where any contingent claim is replicated, plays a key role. In this case, for any $H$, there is $V\in\mathcal{V}_{V_0}$ such that $V=H$. However, real features constraints can make replication unattainable. Such restrictions arise due to liquidity, asset availability, admissibility, or any other trading restrictions as costs. We then have to consider the problem over  $\mathcal{V}\subseteq\mathcal{V}_{V_0}$.

In such a context, agents must bear some risk exposure when
buying/selling contingent claims. More specifically, instead of replicating contingent claims, the goal is to maximize the expected utility for the hedged position. The key ingredient is the utility function $u$, where the usual approach is to solve \[\sup\limits_{V\in\mathcal{V}}E[u(V-H)].\]  Since this is a too complacent approach, especially under uncertainty, it is conventional to consider a robust decision setting where we maximize utility by considering a worst-case scenario over probabilities $\mathcal{Q}$ as \[\sup\limits_{V\in\mathcal{V}}\min\limits_{\mathbb{Q}\in\mathcal{Q}}E_\mathbb{Q}[u(V-H)].\] However, in several situations, using worst-case approaches is too strict for real-world applications. For instance, the agent
may end up with a trivial strategy if the
problem is too punitive.

\cite{Maccheroni2006} then suggest incorporating a penalty term into the objective function, thereby placing a weighting on scenarios that are regarded as more or less relevant. Such a middle-ground approach is related to variational preferences for $u$ under a convex and lower semicontinuous penalty $\alpha\colon\mathcal{Q}\to[0,\infty]$. Thus, the problem then becomes \[\sup\limits_{V\in\mathcal{V}}\min\limits_{\mathbb{Q}\in\mathcal{Q}}\left\lbrace E_\mathbb{Q}[u(V-H)] +\alpha(\mathbb{Q})\right\rbrace.\] 
Both standard expected utility and worst-case approaches are nested in such a model depending on the choice for $\alpha$. See \cite{Hansen2008} for more details and applications. 

Such variational preferences can be represented as the composition of a convex risk measure and a utility function. The theory of risk measures in mathematical finance has become mainstream, especially since the landmark paper of \cite{Artzner1999}. For a comprehensive review, see the books of \cite{Delbaen2012} and \cite{Follmer2016}. A functional $\rho:L^p\rightarrow\mathbb{R}$ is a convex risk measure if it possesses Monotonicity: if $X \leq Y$, then $\rho(X) \geq \rho(Y),\:\forall\: X,Y\in L^p$; Translation Invariance: $\rho(X+c)=\rho(X)-c,\:\forall\: X\in L^p,\:\forall\:c \in \mathbb{R}$; Convexity: $\rho(\lambda X+(1-\lambda)Y)\leq \lambda \rho(X)+(1-\lambda)\rho(Y),\:\forall\: X,Y\in L^p,\:\forall\:\lambda\in[0,1]$. The acceptance set of $\rho$ is defined as $\mathcal{A}_\rho=\left\lbrace X\in L^p:\rho(X)\leq 0 \right\rbrace $. From Theorems 2.11 and 3.1 of \cite{Kaina2009}, $\rho$ is a convex risk measure if and only if it can be represented for any $X\in L^p$ as:
	\begin{equation*}
\rho(X)=\max\limits_{\mathbb{Q}\in\mathcal{Q}}\left\lbrace E_\mathbb{Q}[-X]-\alpha_\rho(\mathbb{Q})\right\rbrace,\:
\alpha_\rho(\mathbb{Q})=\sup\limits_{X\in L^p}\{E_\mathbb{Q}[-X]-\rho(X)\}=\sup\limits_{X\in\mathcal{A}_\rho}E_\mathbb{Q}[-X].
	\end{equation*} Thus, when $\alpha=\alpha_\rho$ for some convex risk measure $\rho$, to maximize the variational preference is equivalent to minimize $\rho_u:=\rho\circ u$ as \[\inf\limits_{V\in\mathcal{V}}\rho_u\left( V-H\right).\]

 We expose a theoretical hedging optimization framework with variational preferences under convex risk measures in this note. In \Cref{dual}, we explore a general dual representation for such compositions concerning penalty terms. Such results allow us, for instance, to compute its sub-differentials as in \Cref{crl:sub}. Further, in \Cref{prp:P}, we study the properties of the optimization problem as a convex and monotone map per se, including its penalty term. We also derive results for optimality based on the obtained sub-differentials, as exposed in \Cref{prp:solution}. The design of the studied optimization problem is related to indifference pricing in \Cref{prp:price}. Further conditions for indifference pricing in the context of the usual fundamental theorems of asset pricing are given in \Cref{thm:price}. We also explore particular examples inside our setup in \Cref{sec:exm}, with connection to usual concrete choices for both risk measure $\rho$ and utility $u$. 
 
 To the best of our knowledge, this is the first work to consider such a hedging problem in the way we do. For risk measures, the seminal quantile hedging approach of \cite{Follmer2000} has opened the stream for risk measures minimization of the hedged position. From there, the main focus is on the minimization of a convex risk measure over the hedged portfolio under some cost constraint, as in  \cite{Nakano2004}, \cite{Barrieu2005}, \cite{Xu2006}, \cite{Rudloff2007},  \cite{Cherny2007}, \cite{Ilhan2009}, \cite{Balbas2010},  \cite{Assa2013}, \cite{Godin2016}, \cite{Cheridito2017}, \cite{Buehler2019},   \cite{Wu2023}. Even non-convex risk measures have been considered, as in \cite{Cong2013} or \cite{Melnikov2022}. However, these papers focus on something other than the problem under a variational preference setup. \cite{Herrmann2017}, for instance, considers variational preferences but is not linked to risk measures.   \cite{Limmer2023} is the closer one, with risk measures and variational preferences considered. However, this work develops different features than we do here since their focus is more on the modeling for a solution.

\section{Main Results}\label{sec:results}

We work with real-valued random variables in a probability space $(\Omega,\mathcal{F},\mathbb{P})$. All equalities and inequalities are considered almost surely in $\mathbb{P}$. Let $L^{p}:=L^{p}(\Omega,\mathcal{F},\mathbb{P})$ be the space of (equivalent classes of) random variables such that  $ \lVert X \rVert_p = (E[|X|^p])^{\frac{1}{p}}<\infty$ for $p\in[1,\infty)$ where $E$ is the expectation operator. We define $L^p_+$ as the cone of non-negative elements of $L^p$. Let $1_A$ be the indicator function for an event  $A$. When not explicit,  definitions and claims are valid for any fixed $L^p,\:p\in[1,\infty)$ with its usual p-norm.  As usual, $L^q$, with $\frac{1}{p}+\frac{1}{q}=1$, is the usual dual of $L^p$. The weak topology on $L^p$, i.e. $\sigma(L^p,L^q)$, is the topology generated by the continuous linear functionals over $L^p$ with the form $X\mapsto E[XY]$ with $Y\in L^q$. Let $\mathcal{Q}$ be the set of all probability measures on $(\Omega,\mathcal{F})$ that are absolutely continuous with respect to $\mathbb{P}$, with Radon--Nikodym derivative $\frac{d\mathbb{Q}}{d\mathbb{P}}\in L^q$. With some abuse of notation, we treat probability measures as elements of $L^q$. 

In what follows, we fix $\rho$ to be a convex risk measure and  $u\colon\mathbb{R}\to\mathbb{R}$ to be a concave and monotone utility. We begin by deducing an explicit form for the dual representation of $\rho_\mu$ based on its penalty term. It is worth noting that without Translation Invariance, the claim for dual representation presented in \Cref{sec:intro}  holds without the need for $E[\mathbb{Q}]=1$ in the proper domain of $\alpha_{\rho_u}$. With some abuse, we call a map monotone if it fulfills the usual monotonicity or the anti-monotonicity version for risk measures. We call $u^*$ the convex conjugate of $-u$.

\begin{Thm} \label{dual}
$\rho_u$ defines a convex, monotone, and continuous map that is normalized when $u(0)=0$. Its penalty term given as  \begin{equation}\label{eq:pen2}
\alpha_{\rho_u}(\mathbb{Q})=\min\limits_{Y\in\mathcal{Q}}\left\lbrace\alpha_\rho(Y)+E_Y[u^*(\mathbb{Q})]\right\rbrace,\:\forall\:\mathbb{Q}\in L^q.   \end{equation}
\end{Thm}

\begin{proof}
It is clear that $\rho_u$ is a convex and monotone map. Further, see \cite{Ruszczynski2006} for instance, a real-valued, convex, monotone functional on a Banach lattice is norm continuous. For each $Y\in\mathcal{Q}$, $G_Y\colon L^p\to\mathbb{R}$ defined as $X\mapsto E_Y[-u(X)]$ also is a convex and monotone map.  As $u$ is continuous, $-u$ can be represented over its convex conjugate $u^*$ as \[-u(x)=\sup\limits_{y\in\mathbb{R}}\left\lbrace -xy-u^*(y)\right\rbrace,\:\forall\:x\in\mathbb{R}.\] We then have that \begin{align*}
    G_Y(X)=E_Y\left[\sup\limits_{q\in\mathbb{R}}\{-qX-u^*(q)\}\right]=\sup\limits_{\mathbb{Q}\in L^q}\left\lbrace E_Y[-X\mathbb{Q}]-E_Y[u^*(\mathbb{Q})]\right\rbrace.
\end{align*}
The interchange between supremum and expectation is due to Theorem 14.60 in \cite{Rockafellar2009}, which is valid since $L^p$ spaces are decomposable, i.e. for any $X\in L^p$, $A\in\mathcal{F}$ and $W\in L^\infty$, $1_A X+1_{A^c}W\in L^p$. Thus, the penalty for each $G_Y$ is given as \[\alpha_{G_Y}(\mathbb{Q})=\sup_{X\in L^p}\left\lbrace E_Y[-X\mathbb{Q}]-G_Y(X)\right\rbrace=E_Y[u^*(\mathbb{Q})],\:\forall\:\mathbb{Q}\in L^q.\]  Further, the maximum in the representation for $\rho$ can be taken over the weakly compact $\mathcal{Q}^\prime:=\{\mathbb{Q}\in L^q\colon\alpha_\rho(\mathbb{Q})<\infty\}$.  For any $\mathbb{Q}\in L^q$ we have that \begin{align*}
\alpha_{\rho_u}(\mathbb{Q})&=\sup\limits_{X\in L^p}\left\lbrace E[-X\mathbb{Q}]+\min\limits_{Y\in\mathcal{Q}^\prime}\left\lbrace E_Y[u(X)]+\alpha_\rho(Y)\right\rbrace \right\rbrace\\
&=\min\limits_{Y\in\mathcal{Q}^\prime}\left\lbrace\sup\limits_{X\in L^p}\left\lbrace E[-X\mathbb{Q}] +E_Y[u(X)]\right\rbrace +\alpha_\rho(Y) \right\rbrace\\
&=\min\limits_{Y\in\mathcal{Q}^\prime}\left\lbrace\alpha_\rho(Y)+E_Y[u^*(\mathbb{Q})]\right\rbrace.
\end{align*}
The second inequality is due to Sion minimax theorem, see \cite{Sion1958}, since the map $(X\times Y)\mapsto E[-X\mathbb{Q}]+ E_Y[u(X)]+\alpha_{\rho}(Y)$ is concave and upper semicontinuous in the first argument, over the convex set $L^p$. In contrast, it is convex and weak lower semicontinuous in the second argument, with $\mathcal{Q}^\prime$ weak compact. The claim follows since the minimum is not altered when taken on the whole $\mathcal{Q}$. This concludes the proof.
\end{proof}

\begin{Rmk}
It is straightforward to verify that the supremum for the dual representation of $\rho_u$ can be taken only on those $\mathbb{Q}\geq 0$, in order to keep monotonicity. Moreover, such penalty term keeps some similarity to the one for robust convex risk measures in \cite{Righi2024b}, where the family $\mathbb{Q}\mapsto E_Y[u^*(\mathbb{Q})],\:Y\in\mathcal{Q}$ plays the role for the penalty of their auxiliary maps $g_Y$.
\end{Rmk}

In convex analysis, sub-differentials play a critical role in optimization. For any $f\colon L^p\to\mathbb{R}$, its sub-gradient at $X\in L^p$ is $\partial f(X)=\{Y\in L^q\colon \rho(Z)-\rho(X) \geq E[(Z-X)Y]\:\forall\:Z\in L^p\}$. For a convex risk measure $\rho$, Theorem 3 of \cite{Ruszczynski2006} assures that
       \[\partial \rho(X)=\left\lbrace \mathbb{Q}\in\mathcal{Q}\colon \rho(X)= E_\mathbb{Q}[-X]-\alpha_\rho(\mathbb{Q})\right\rbrace\neq\emptyset.\]  We say $f\colon L^p\to\mathbb{R}$ is Gâteaux differentiable at $X\in L^p$ when $t\mapsto\rho(X+tZ)$ is differentiable at
$t = 0$ for any $Z\in L^p$ and the derivative defines a continuous linear functional on $L^p$. Furthermore,
     $\rho$ is Gâteaux differentiable at $X$ if and only if $\partial \rho(X)=\{\mathbb{Q}\}$, which in this case the derivative turns out to be  $\mathbb{Q}$, i.e. the map $Z\mapsto E_\mathbb{Q}[-Z]$. 

\begin{Crl}\label{crl:sub}
We have that
\begin{equation*}\label{eq:subdif2}
\partial \rho_u(X) = \bigcup\limits_{Y\in\partial\rho(u(X))}  \partial E_Y[-u(X)],\:\forall\:X\in L^p.  
\end{equation*}
In particular, $\rho_u$ is Gâteaux differentiable at $X$ if and only if $\rho$ is Gâteaux differentiable at $u(X)$.
\end{Crl}

\begin{proof}
For fixed $\mathbb{Q}\in L^q_+$, let $Y$ denote the argmin of \eqref{eq:pen2} with respect to $\mathbb{Q}$. We then have by \Cref{dual} that \begin{align*}
    \mathbb{Q}\in\partial \rho_u(X)&\iff E_\mathbb{Q}[-X]-\alpha_\rho(Y)-E_Y[u^*(\mathbb{Q})]=\rho(u(X))\\
    &\iff E_Y[-u(X)]-\alpha_\rho(Y)=\rho(u(X))\:\textbf{and}\: E_\mathbb{Q}[-X]-E_Y[u^*(\mathbb{Q})]=E_Y[-u(X)]\\
    &\iff Y\in\partial\rho(u(X)) \:\textbf{and}\: \mathbb{Q}\in \partial E_Y[-u(X)]\\
    &\iff \mathbb{Q}\in\bigcup\limits_{Y\in\partial\rho(u(X))}  \partial E_Y[-u(X)].  
    \end{align*} 
    The claim for the Gâteaux derivative is straightforwardly obtained.
\end{proof}

\begin{Rmk}
 For $L^\infty$, the space of essentially bounded random variables, the results hold when we consider the dual pair $(L^\infty,L^1)$ with its weak* topology, and assume $\rho$ to be Lebesgue continuous, i.e., continuous regarding $\mathbb{P}-a.s.$ dominated convergence. See Theorem 4.33 in \cite{Follmer2016} for the dual representation, and Theorem 21 and Proposition 14 in \cite{Delbaen2012} for sub-differentials. In fact, Examples 3 and 5 in \Cref{sec:exm} remain valid with the same calculations, but with Gateaux derivatives in $L^1$.
\end{Rmk}

We now turn to the formal definition of the hedging optimization problem. With that, we are able to prove its properties and conditions for solution. The value obtained in the solution refers to the minimal amount of cash that, when injected into the utility of the hedged position, turns it acceptable from a risk measure point of view, provided the
the position is hedged optimally.

\begin{Def}
Let $H$ be a contingent claim, $\mathcal{V}\subseteq \mathcal{V}_{V_0}$ non-empty and convex, and $\inf_{V\in\mathcal{V}}\rho_u\left( V\right)$ finite. Then, the  hedging problem is defined as \begin{equation*}
    P (H)\::\:\inf\limits_{V\in\mathcal{V}}\rho_u\left( V-H\right).
\end{equation*}
\end{Def}

\begin{Rmk}
The absence of arbitrage is not a requirement in our approach since we only have the no irrelevance assumption given as $P(0)=\inf_{V\in\mathcal{V}}\rho_u\left( V\right)>-\infty$. It is important to note that market irrelevance is not only implied by explicit arbitrage opportunities. Statistical arbitrage can also manifest itself in markets where arbitrage is not apparent. 
\end{Rmk}

\begin{Thm}\label{prp:P}
The map $H\mapsto P(H)$ is finite, monotone, convex and continuous. Its penalty term is given as \begin{equation*}
\alpha_{P}(\mathbb{Q})=\alpha_{\rho_u}(\mathbb{Q})+\sup\limits_{V\in\mathcal{V}}E_\mathbb{Q}[V],\:\forall\:\mathbb{Q}\in\mathcal{Q}.\end{equation*}
\end{Thm}

\begin{proof}
Monotonicity is straightforward. For convexity, we have for any $H_1,H_2$ and any $\lambda\in[0,1]$ that \begin{align*}
P(\lambda H_1+(1-\lambda)H_2)&=\inf\limits_{V_1,V_2\in\mathcal{V}}\rho_u\left( \lambda(V_1-H_1)+(1-\lambda)(V_2-H_2)\right)\\
&\leq\inf\limits_{V_1,V_2\in\mathcal{V}}\left\lbrace\lambda\rho_u\left( V_1-H_1\right)+(1-\lambda)\rho_u\left( V_2-H_2\right)\right\rbrace\\
&=\lambda P(H_1)+(1-\lambda)P(H_2).
\end{align*}
Moreover, since $P<\infty$ and $P(0)>-\infty$, is a well-known fact from convex
analysis that $P$ is finite. Thus, similarly as $\rho_u$, we have that $P$ is continuous. Then, $P$ can be represented over $\alpha_P$ with the due sign change. Let $\mathbb{I}_A$ be the convex characteristic function of $A\subseteq L^q$, i.e. $\mathbb{I}_A(\mathbb{Q})=0$ if $\mathbb{Q}\in A$, and $\mathbb{I}_A(\mathbb{Q})=\infty$, otherwise. Thus, we have for any $\mathbb{Q}\in L^q$ that \begin{align*}
\alpha_P(\mathbb{Q})&=\sup\limits_{H\in L^p,V\in\mathcal{V}}\left\lbrace E[H\mathbb{Q}]-\max\limits_{Y\in\mathcal{Q}^\prime}\left\lbrace E_Y[H-V]-\alpha_{\rho_u}(Y)\right\rbrace\right\rbrace\\
&=\min\limits_{Y\in\mathcal{Q}^\prime}\sup\limits_{H\in L^p,V\in\mathcal{V}}\left\lbrace E[H\mathbb{Q}]+ E_Y[V-H]+\alpha_{\rho_u}(Y)\right\rbrace\\
&=\min\limits_{Y\in\mathcal{Q}^\prime}\left\lbrace\sup\limits_{H\in L^p} \left\lbrace E[H\mathbb{Q}]- E_Y[H]\right\rbrace+\alpha_{\rho_u}(Y)+\sup\limits_{V\in\mathcal{V}}E_Y[V]\right\rbrace\\
&=\min\limits_{Y\in\mathcal{Q}^\prime}\left\lbrace\mathbb{I}_{\{Y\}}(\mathbb{Q})+\alpha_{\rho_u}(Y)+\sup\limits_{V\in\mathcal{V}}E_Y[V]\right\rbrace\\
&=\alpha_{\rho_u}(\mathbb{Q})+\sup\limits_{V\in\mathcal{V}}E_\mathbb{Q}[V].\end{align*}
The second inequality is due to Sion minimax Theorem, since the map $(H\times V,Y)\mapsto E[H\mathbb{Q}]+ E_Y[V-H]+\alpha_{\rho_u}(Y)$ is convex and lower semicontinuous in the first argument, over the convex set $L^p\times\mathcal{V}$. In contrast, it is concave and weak upper semicontinuous in the second argument over the weak compact $\mathcal{Q}^\prime$. 
\end{proof}

From the convexity of both $\rho_u$ and $\mathcal{V}$, it is clear that a necessary and sufficient condition of optimality is $\nabla_h\rho_u(V-H)=0$. Thus, obtaining such a derivative is crucial for a solution. We now relate the Gâteaux derivative of $\rho_u$ with the optimality of the hedging problem. Recall that the normal cone of a convex set $A\subseteq\mathbb{R}^n$ at $x\in A$ is $N_A(x)=\{y\in\mathbb{R}^n\colon (z-x)^Ty\leq0\:\forall\:z\in A\}$.

\begin{Thm}\label{prp:solution}
 Let $\rho_u$ be Gâteaux differentiable. Then, $V$ is a solution for $P(H)$ if and only if \[E_{\mathbb{Q}^{V}}[\Delta S_i]\in N_\mathcal{V}(V)\:\forall\:i=1,\dots,n,\:\mathbb{Q}^{V}\in\partial \rho_u(V-H).\] 
\end{Thm}

\begin{proof}
Let $f\colon\mathbb{R}^n\to L^p$ be given as $f(h)=V_0+\sum_{i=1}^n h_i\Delta S_i-H$. Then, linearity implies that $f(\cdot)(\omega)$ is differentiable for  $\mathbb{P}-a.s.$ every $\omega$. This implies that it is also differentiable for  $\mathbb{Q}-a.s.$ every $\omega$ for any $\mathbb{Q}\in\mathcal{Q}$. More specifically, $\nabla_h  f(V)=(\Delta S_1,\dots,\Delta S_n)$ for any $V\in\mathcal{V}$. Thus, Proposition 1 of \cite{Ruszczynski2006} assures that $\rho_u\circ f$ is differentiable with $\nabla_h  (\rho_u\circ f)(x)=-E[\nabla_h  f(x)\mathbb{Q}]$, where $\mathbb{Q}\in\partial \rho_u(f(x))$. Further, Propositions 2 and 3 of \cite{Ruszczynski2006} states that \[x\in\argmin\limits_{y\in A}\rho_u(f(y))\:\iff\:E[\nabla_h  f(x)\mathbb{Q}]\in N_A(x),\]
   where $A\subseteq\mathbb{R}^n$ is convex. Thus, $V$ is a solution if and only if $E[\nabla_h  f(V)\mathbb{Q}^{V}]\in N_\mathcal{V}(V)$, where $\mathbb{Q}^V\in\partial\rho_u(V-H)$. 
\end{proof}

\begin{Rmk}
When $\mathcal{V}=\mathcal{V}_{V_0}$, the condition in \Cref{prp:solution} simplifies to $E_{\mathbb{Q}^{V}}[\Delta S_i]=0$ for any $i=1,\dots,n$ since $\mathcal{V}_{V_0}$ is isomorphic to $\mathbb{R}^n$, which implies $N_{\mathcal{V}_{V_0}}(V)=0$ for any $V$. If $\mathcal{V}$ is compact (closed and bounded) and convex, then $P(H)$ has a solution since the map $f\colon V\mapsto\rho_u(V-H)$ is continuous. Moreover, as $\mathcal{V}$ is compact, the argmin is bounded. Furthermore, convexity and continuity of $f$ imply that the argmin is closed and convex, hence compact.    
\end{Rmk}

From $P(H)$, one can determine the indifference pricing of $H$. To that, we have to normalize $P(H)$ by the capital required for hedging a zero claim. We now define the prices for both seller and buyer.

\begin{Def}
The seller and buyer prices of a contingent claim $H$ are given, respectively, as \begin{align*}
SP(H)=P(H)-P(0)\:\text{and}\:BP(H)=-P(-H)+P(0).
\end{align*}
\end{Def}

We now relate such indifference prices with the fundamental theorems of asset pricing. Recall that set of equivalent martingale measures, $EMM\subset\mathcal{Q}$, are the probability measures equivalent to $\mathbb{P}$ such that $E_\mathbb{Q}[\Delta S_i]=0$ for any $i=1,\dots,n$. It is well known that the market is free of arbitrage if and only if $EMM\neq\emptyset$, while it is complete if and only if $EMM$ is a singleton.

\begin{Lmm}\label{prp:price} 
$SP(H)\geq BP(H)\geq 0$. If $H$ is attainable and $\mathcal{V}=\mathcal{V}_{V_0}$, then $SP(H)=SP(V_0)$ and $BP(H)=BP(V_0)$.   
\end{Lmm}

\begin{proof}
Since $P$ is convex, we have \[P(0)= P\left(\frac{1}{2}(H-H)\right)\leq \frac{1}{2}(P(H)+P(-H)).\] Then, $SP(H)\geq BP(H)$. Moreover, by monotonicity of $P$, we have that $P(-H)\leq P(0)$. Hence, $BP(H)\geq 0$. When $\mathcal{V}=\mathcal{V}_{V_0}$, if $H$ is attainable there is some $V^*$ such as $H=V_0+h^*\Delta S$. We thus obtain\begin{align*}
P(H)&=\inf\limits_{h\in\mathbb{R}^n}\rho_u\left( V_0+h^\prime\Delta S-(V_0+h^*\Delta S)\right)\\
&=\inf\limits_{h\in\mathbb{R}^n}\rho_u\left( V_0+(h^\prime-h^*)\Delta S-V_0\right)\\
&=\inf\limits_{h\in\mathbb{R}^n}\rho_u\left( V_0+h^\prime\Delta S-V_0\right)=P(V_0).
\end{align*}
The last inequality follows since $\mathbb{R}^n-h^*=\mathbb{R}^n$. Hence, $SP(H)=SP(V_0)$. The claim for $BP$ is analogously obtained.
\end{proof}

\begin{Thm}\label{thm:price} 
Let $\rho_u(X)\leq -\essinf X$ for any $X\in L^p$ and $\mathcal{V}=\mathcal{V}_{V_0}$. If the market is free of arbitrage, then \[\sup\limits_{\mathbb{Q}\in EMM}E_\mathbb{Q}[H]\geq SP(H)\geq BP(H)\geq\inf\limits_{\mathbb{Q}\in EMM}E_\mathbb{Q}[H].\] If, in addition, the market is complete, then $SP(H)=BP(H)=SP(V_0)=BP(V_0)$.  
\end{Thm}

\begin{proof}
Let the market be free of arbitrage. In this case, it is well known, see Theorem 1.32 in \cite{Follmer2016} for instance, that the super-hedging and sub-hedging prices are given, respectively, as \begin{align*}
\inf\{x\in\mathbb{R}\colon \exists\:h\in\mathcal{V}\:\text{s.t.}\:x+h^\prime\Delta S\geq H\}&=\inf\limits_{h\in\mathcal{V}}\esssup(H-h^\prime\Delta S)=\sup \limits_{\mathbb{Q}\in EMM}E_\mathbb{Q}[H],\\
\sup\{x\in\mathbb{R}\colon \exists\:h\in\mathcal{V}\:\text{s.t.}\:H\geq x+h^\prime\Delta S\}&=\sup\limits_{h\in\mathcal{V}}\essinf(H-h^\prime\Delta S)=\inf \limits_{\mathbb{Q}\in EMM}E_\mathbb{Q}[H].
\end{align*}
It is clear that \begin{align*}
P(H)+V_0&=\inf\limits_{h\in\mathcal{V}}\rho_u(V_0 +h^\prime\Delta S-H)+V_0\leq \inf\limits_{h\in\mathcal{V}}\esssup(H-h^\prime\Delta S),\\
-P(-H)-V_0&=-\inf\limits_{h\in\mathcal{V}}\rho_u(V_0 +h^\prime\Delta S+H)-V_0\geq \sup\limits_{h\in\mathcal{V}}\essinf(H-h^\prime\Delta S).
\end{align*}  We now claim that in this context $V_0=-P(0)$. Notice that any constant contingent claim $H=C$ is attainable since $C=\inf_{\mathbb{Q}\in EMM}E_\mathbb{Q}[C]=\sup _{\mathbb{Q}\in EMM}E_\mathbb{Q}[C]$. Then, both $V_0$ and $-V_0$ are attainable. Thus, by \Cref{prp:price}, we have that $P(V_0)=P(-V_0)$. Since $0$ also is attainable, we have that $0=SP(0)=SP(V_0)=P(V_0)-P(0)$. Thus, $P(V_0)=P(0)$. Further, by considering $H=V_0$, we have that \[-P(-V_0)-V_0\geq\inf\limits_{\mathbb{Q}\in EMM}E_\mathbb{Q}[V_0]=V_0=\sup\limits_{\mathbb{Q}\in EMM}E_\mathbb{Q}[V_0]\geq P(V_0)+V_0.\] Thus, $-2V_0=P(V_0)+P(-V_0)=2P(0)$. Hence, we have that both $SP(H)\leq\sup_{\mathbb{Q}\in EMM}E_\mathbb{Q}[H]$ and $BP(H)\geq\inf_{\mathbb{Q}\in EMM}E_\mathbb{Q}[H]$. The claim now follows by \Cref{prp:price} since $SP(H)\geq BP(H)$. If the market is complete, then $EMM=\{\mathbb{Q}^*\}$. In this case the price of $H$ is given as $E_{\mathbb{Q}^*}[H]=\inf_{\mathbb{Q}\in EMM}E_\mathbb{Q}[H]=\sup_{\mathbb{Q}\in EMM}E_\mathbb{Q}[H]$. Thus, we directly have $SP(H)=BP(H)$. This concludes the proof. 
\end{proof}

\section{Examples}\label{sec:exm}

 The minimization problem that defines $P$ can be equivalently formulated regarding returns or price variations. Let $\Delta H = H-V_0$. Then \[\rho_u(V-H)=\rho_u\left( V_0+h^\prime\Delta S-(V_0+\Delta H)\right)=\rho_u\left( h^\prime\Delta S-\Delta H\right).\] In this case, we can directly interpret $h$ as weighting schemes instead of numbers of shares. Thus, in the following examples we consider $\mathcal{V}=\{h\in\mathbb{R}^n\colon\sum_{i=1}^n h_i=1\}$. It is non-empty, closed, and convex despite not being bounded. One can also consider only its positive part, i.e., the one with a short-selling restriction as $\{h\in\mathbb{R}^n\colon\sum_{i=1}^n h_i=1,\:h\geq 0\}$.

 \begin{Exm}
    Let $\rho(X)=-E[X]$ and $u(x)=-e^{-a x}+1,\:a>0$. Consider the special case where $\Delta S-\Delta H$ is multivariate normally distributed with vector of means $\mu=E[ (\Delta S-\Delta H)]$ and covariance matrix $\Sigma=\Sigma(\Delta S-\Delta H)$. For each $h\in\mathcal{V}$ we have that $e^{-a(h^\prime\Delta S-\Delta H)}$ has a log-normal distribution. Denote $\mathbf{1}$ the constant vector 1 in $\mathbb{R}^n$. Then the optimal solution to $P(H)$ is \[h=\frac{\lambda\Sigma^{-1}\mathbf{1}+\Sigma^{-1}\mu}{a},\:\text{where}\:\lambda=\frac{a-\mathbf{1}^\prime \Sigma^{-1}\mu}{\mathbf{1}^\prime \Sigma^{-1}\mathbf{1}}.\]  The claim follows because the optimal hedging problem becomes \begin{align*}
   & \inf\limits_{\{h\in\mathbb{R}^n\colon h^\prime\mathbf{1}=1\}}\left\lbrace\exp\left(a\left(-h^\prime\mu+\frac{a h^\prime \Sigma h}{2}\right)\right)-1\right\rbrace.
\end{align*}
It is straightforward to observe that $h^*$ is a solution to this problem if and only if it is a solution to \[\inf\limits_{\{h\in\mathbb{R}^n\colon h^\prime\mathbf{1}=1\}}\left\lbrace -h^\prime \mu+\frac{a h^\prime\Sigma h}{2}\right\rbrace.\] Thus, from convexity in $h$, the solution is given through the following Lagrangian multipliers equation system \[\begin{cases}\nabla_h \left(-h^\prime \mu+\dfrac{a h^\prime\Sigma h}{2}\right)\\
+\lambda \nabla_h (\mathbf{1}^\prime h-1)=0\\
h^\prime \mathbf{1}=1\end{cases}\iff\begin{cases}-\mu+a\Sigma h-\lambda \mathbf{1}=0\\
h^\prime \mathbf{1}=1\end{cases}\iff\begin{cases}h=\Sigma^{-1}\dfrac{(\lambda \mathbf{1}+\mu)}{a}\\
1=\mathbf{1}^\prime h\end{cases}.\]
By solving this system we get that  \[h=\frac{\lambda\Sigma^{-1}\mathbf{1}+\Sigma^{-1}\mu}{a}\:\text{and}\:\lambda=\frac{a-\mathbf{1}^\prime \Sigma^{-1}\mu}{\mathbf{1}^\prime \Sigma^{-1}\mathbf{1}}. \]
\end{Exm}

\begin{Exm}
The choice for the exponential utility map from the last example is related to the entropic risk measure (ENT), which is a map that depends on the user's risk aversion through such an exponential utility function. Formally, it is the map $ENT^a\colon L^1\to\mathbb{R}$  defined as \[ENT^a (X) = \frac{1}{a} \log \E [ e^{-aX}],\:a>0.\]  Let $u(x)=x$, and again $\Delta S-\Delta H$ follows a multivariate normal distributed with parameters $\mu=E[ (\Delta S-\Delta H)]$ and $\Sigma=\Sigma (\Delta S-\Delta H)$. The optimization problem then becomes 
\begin{align*}
&\inf\limits_{\{h\in\mathbb{R}^n\colon h^\prime\mathbf{1}=1\}}\left\lbrace\dfrac{1}{a}\log\left(\exp\left(a\left(-h^\prime\mu+\frac{a h^\prime \Sigma h}{2}\right)\right)\right)\right\rbrace\\
=&\inf\limits_{\{h\in\mathbb{R}^n\colon h^\prime\mathbf{1}=1\}}\left\lbrace -h^\prime \mu+\frac{a h^\prime\Sigma h}{2}\right\rbrace.
\end{align*} Thus, the solution is the same as the one in the previous example.
\end{Exm}

\begin{Exm}
Let $F_{X}(x) = P(X\leq x)$ and $F_{X}^{-1}(\alpha)=\inf\{x\in\mathbb{R}\colon F_X (x)\geq\alpha\}$ for $\alpha \in (0,1)$ be, respectively, the distribution function and the (left) quantile of $X$. The Value at Risk (VaR), which is defined as $VaR^\alpha(X)=-F^{-1}_X(\alpha)$, $\alpha\in(0,1)$, is the most prominent example of tail risk measure, despite not being convex. It generates the canonical example for a convex tail risk measure, the Expected Shortfall (ES), that is functional $ES^\alpha\colon L^1\to\mathbb{R}$ defined as \[ES^\alpha(X)=\frac{1}{\alpha}\int_0^\alpha VaR^udu,\:\alpha\in(0,1).\] The ES is positive homogeneous in the sense that $ES^\alpha(\lambda X)=\lambda ES^\alpha(X)$ for any $X\in L^1$ and any $\lambda\geq0$. Moreover, it is Gâteaux differentiable at any $X\in L^1$ with derivative $\frac{1}{\alpha}1_{X\leq F^{-1}_X(\alpha)}$. Let $\rho=ES^\alpha$ and $u(x)=x$. Thus,  we have that \[\nabla_h  ES^\alpha(h^\prime(\Delta S-\Delta H))=-\frac{1}{\alpha}E\left[\Delta S_i1_{\left\lbrace h^\prime \Delta S -\Delta H\leq F_{ \{h^\prime \Delta S-\Delta H\}}^{-1}(\alpha)\right\rbrace}\right],i=1,\dots,n.\] Thus, under \Cref{prp:solution}, the optimality condition turns to the solution of the following system of equations \[\begin{cases}E\left[\Delta S_i 1_{\left\lbrace h^\prime \Delta S -\Delta H\leq F_{ \{h^\prime \Delta S-\Delta H\}}^{-1}(\alpha)\right\rbrace}\right]=-\alpha\lambda,\:\forall\:i=1,\dots,n\\
h^\prime\mathbf{1}=1\end{cases}.\]  Notice that the optimality condition is not altered if we consider the affine utility function $u(x)=a+bx$, where $a,b>0$, due to Translation Invariance and the Positive Homogeneity of $ES^\alpha$.
\end{Exm}

\begin{Exm}
For this example, consider again $\rho=ES^\alpha$ and $u(x)=-e^{-a x}+1,\:a>0$, under the  case where $\Delta S-\Delta H$ follows a multivariate normal distribution with parameters $\mu=E[\Delta S-\Delta H]$ and $\Sigma=\Sigma (\Delta S-\Delta H)$. Let $\Phi$ and $\phi$ be, respectively, the standard Normal cumulative and density probability distributions. The Expected shortfall for such a random variable with this distribution of probability is well known to be \[ES^\alpha(h^\prime \Delta S-\Delta H)=\varphi(h)=-\exp\left(-a h' \mu + \frac{a^2 h' \Sigma h}{2}\right) \frac{\Phi\left(\Phi^{-1}(\alpha) - a \sqrt{h' \Sigma h}\right)}{\alpha}.\] Thus, the  problem $P(H)$ becomes 
\[
\inf\limits_{\{h\in\mathbb{R}^n\colon h^\prime\mathbf{1}=1\}}\varphi(h)=\inf\limits_{\{h\in\mathbb{R}^n\colon h^\prime\mathbf{1}=1\}}\left\lbrace-\exp\left(-a h' \mu + \frac{a^2 h' \Sigma h}{2}\right) \frac{\Phi\left(\Phi^{-1}(\alpha) - a \sqrt{h' \Sigma h}\right)}{\alpha}\right\rbrace
\]
We obtain the following expression using multiple chain rule and matrix calculus steps.
\[
\nabla_h  \varphi(h) = \frac{\mathcal{E}(h)}{\alpha} \left[ \left( a \mu - a^2 \Sigma h \right) \Phi\left(\Phi^{-1}(\alpha) - a \sqrt{h' \Sigma h}\right) + a \phi\left(\Phi^{-1}(\alpha) - a \sqrt{h' \Sigma h}\right) \frac{\Sigma h}{\sqrt{h' \Sigma h}} \right],\] where \[
\mathcal{E}(h) = \exp\left(-a h' \mu + \frac{a^2 h' \Sigma h}{2}\right).
\] Thus, we obtain the solution to $P(h)$ by solving the system \[\begin{cases}\nabla_h  \varphi(h)=\lambda\mathbf{1}\\
h^\prime \mathbf{1}=1\end{cases}.\] Such a solution can be obtained through usual numerical solutions.

\end{Exm}

\begin{Exm}
A loss function that is very used for hedging optimization is $u(x)=-x^{-}=\min\{x,0\}$. This map is considered, for instance, in \cite{Nakano2004} and \cite{Rudloff2007}, where the focus is on the shortfall risk of the hedged position. In this case the optimization problem becomes $\inf_{V\in\mathcal{V}}\rho_u\left((V-H)^-\right)$. Since in this case $u^*(q)=\mathbb{I}_{[0,1]}(q)$, we have by \Cref{dual} that the penalty term becomes \[\alpha_{\rho_u}(\mathbb{Q})=\min\left\lbrace\alpha_\rho(Y)\colon Y(\mathbb{Q}\in[0,1])=1\right\rbrace.\] Furthermore, the sub-differential of $X\mapsto E_Y[-u(X)]$ is obtained as $\mathbb{Q}=Y1_{X\leq 0}$. Thus, in accordance to \Cref{crl:sub}, we have that \[\partial \rho_u(X)=\{Y1_{X\leq 0}\colon Y\in\partial\rho(-X^-)\}.\] In particular, if $\rho$
 is Gateaux differentiable at $u(X)$ with derivative $\mathbb{Q}$,  by \Cref{prp:solution} $h$ is a solution for $P(H)$ if and only if it solves the systems of equations given as \[\begin{cases}E_\mathbb{Q}[\Delta S_i1_{h^\prime\Delta S\leq \Delta H}]=\lambda,\:\forall\:i=1,\dots,n\\
 h^\prime\mathbf{1}=1\end{cases}.\]
 \end{Exm}

\section*{Disclosure of interest}

 There are no interests to declare.

\bibliography{Theory,reference}
\bibliographystyle{apalike}

\end{document}